\documentclass[letterpaper,10pt,onecolumn,conference]{IEEEtran}

\IEEEoverridecommandlockouts                              

\title{\LARGE \bf
Directly Coupled 
Observers for Quantum Harmonic Oscillators with Discounted Mean Square Cost  Functionals 
and Penalized Back-action$^*$}

\usepackage{datetime}

\author{Igor G. Vladimirov$^{\dagger}$, \qquad Ian R. Petersen$^{\dagger}$
\thanks{$^*$This work is supported by the Australian Research Council.}
\thanks{$^\dagger$UNSW Canberra, Australia.
{\tt igor.g.vladimirov@gmail.com}, {\tt i.r.petersen@gmail.com}.}
}

\usepackage{amsmath}
\usepackage{amssymb}
\usepackage{amsfonts}
\usepackage{mathptmx} 
\usepackage{times}
\usepackage{graphicx}

\DeclareMathAlphabet{\bit}{OML}{cmm}{b}{it}
\newtheorem{lem}{Lemma}
\newtheorem{thm}{Theorem}

\def\ad{\mathrm{ad}}           

\def\<{\leqslant}           
\def\>{\geqslant}           

\def\d{\partial}

\def\wt{\widetilde}

\def\Re{\mathrm{Re}}   
\def\Im{\mathrm{Im}}   

\def\cH{\mathcal{H}}   
\def\mA{\mathbb{A}}    
\def\mR{\mathbb{R}}    
\def\mC{\mathbb{C}}    

\def\Tr{\mathrm{Tr}}       
\def\rT{\mathrm{T}}        
\def\bS{\mathbf{S}}

\def\bE{\mathbf{E}}    

\def\rprod{\mathop{\overrightarrow{\prod}}}

\def\bra{{\langle}}
\def\ket{{\rangle}}

\def\Bra{\left\langle}
\def\Ket{\right\rangle}

\def\re{\mathrm{e}}        
\def\rd{\mathrm{d}}        

\def\sj{\mathsf{j}}

\def\br{\mathbf{r}}
\def\x{\times}
\def\ox{\otimes}

\def\cZ{{\mathcal Z}}

\def\cX{\mathcal{X}}

\def\cK{\mathcal{K}}

\def\cC{\mathcal{C}}

\def\cI{\mathcal{I}}
\def\cP{\mathcal{P}}
\def\cQ{\mathcal{Q}}

\def\cA{\mathcal{A}}

\def\cE{\mathcal{E}}

\def\bL{\mathbf{L}}

\def\mH{\mathbb{H}}
\def\mS{\mathbb{S}}

\def\diag{\mathop{\mathrm{diag}}}    

\begin{document}

\maketitle
\thispagestyle{empty}
\pagestyle{plain}

\begin{abstract}
This paper is concerned with quantum harmonic oscillators consisting of a quantum plant and a directly coupled coherent quantum observer. We employ discounted quadratic performance criteria in the form of exponentially weighted time averages of second-order moments of the system variables. A coherent quantum filtering (CQF) problem is formulated as the minimization of the discounted mean square of an estimation error, with which the dynamic variables of the observer approximate those of the plant. The cost functional also involves a quadratic penalty on the plant-observer coupling matrix in order to mitigate the back-action of the observer on the covariance dynamics of the plant. For the discounted mean square optimal CQF problem with penalized back-action, we establish first-order necessary conditions of optimality in the form of algebraic matrix  equations. By using the Hamiltonian structure of the Heisenberg dynamics and related Lie-algebraic techniques, we represent this set of equations in a more explicit form in the case of equally dimensioned plant and observer.
\end{abstract}

\section{INTRODUCTION}

Noncommutative counterparts of classical control and filtering problems \cite{AM_1989,KS_1972} are a subject of active research in quantum control which is concerned with  dynamical and stochastic systems governed by the laws of quantum mechanics and quantum probability \cite{H_2001,M_1995}. These developments (see, for example, \cite{JNP_2008,MJ_2012,NJP_2009,VP_2013a,VP_2013b}) are particularly focused on open quantum systems whose internal dynamics are affected by interaction with the environment. In such systems, the evolution of dynamic variables (as noncommutative operators on a Hilbert space) is often modelled using the Hudson-Parthasarathy calculus \cite{H_1991,HP_1984,P_1992} which provides a rigorous footing for quantum stochastic  differential equations (QSDEs) driven by quantum Wiener processes on symmetric Fock spaces.
In particular, linear QSDEs model open quantum harmonic oscillators (OQHOs) \cite{EB_2005} whose dynamic  variables (such as the position and momentum or annihilation and creation operators \cite{M_1998,S_1994}) satisfy canonical commutation relations (CCRs). This class of QSDEs is important for linear quantum control theory \cite{P_2010} and applications to quantum optics \cite{GZ_2004,WM_1994_book}.

One of the fundamental problems for quantum stochastic systems is the coherent quantum linear quadratic Gaussian (CQLQG)  control problem \cite{NJP_2009} which is a quantum mechanical counterpart of the classical LQG control problem. The latter is well-known in linear stochastic control theory due to the separation principle and its links with Kalman filtering and deterministic optimal control settings such as the linear quadratic regulator (LQR) problem \cite{AM_1989,KS_1972}. An important part of this theory is the stochastic filtering theory which has its roots in the works of Kolmogorov  and Wiener of the 1940s \cite{K_1941,W_1949} and is concerned with estimating a random process of interest by using the past history of measurements of another random process. However, coherent quantum feedback control \cite{L_2000,WM_1994_paper} employs the idea of control by interconnection, whereby quantum systems interact with each other directly or through optical fields in a measurement-free fashion, which can be described using the quantum feedback network formalism \cite{GJ_2009}. In comparison with the traditional observation-actuation control paradigm, coherent quantum control avoids the ``lossy'' conversion  of operator-valued quantum variables into classical signals (which underlies the quantum measurement process), is potentially faster and can be implemented on micro and nano-scales using natural  quantum mechanical effects.

In coherent quantum filtering (CQF) problems \cite{MJ_2012,VP_2013b}, which are ``feedback-free'' versions of the CQLQG control problem, an observer is cascaded in a measurement-free fashion with a quantum plant so as to develop quantum correlations with the latter over the course of time. Both problems employ mean square performance criteria and involve physical realizability (PR) constraints \cite{JNP_2008} on the state-space matrices of the quantum controllers and filters. The PR constraints are a consequence of the specific Hamiltonian structure of quantum dynamics and complicate the  design of optimal coherent quantum controllers and filters. Variational approaches of \cite{VP_2011b}--\cite{VP_2013b}  reformulate the underlying problem as a constrained covariance control problem and employ an adaptation of ideas from dynamic programming, the Pontryagin minimum principle \cite{PBGM_1962,SW_1997} and nonlinear functional analysis. In particular, the Frechet differentiation of the LQG cost with respect to the state-space matrices of the controller or filter  subject to the PR constraints leads to necessary conditions of optimality in the form of nonlinear algebraic matrix equations. Although this approach is quite similar to \cite{BH_1998,SIG_1998} (with the quantum nature of the problem manifesting itself only through the PR constraints), the resulting equations appear to be much harder to solve than their classical predecessors.

In a recent work \cite{V_2015a}, a methodological shift has been undertaken towards fully quantum variational techniques  based on infinitesimal perturbation analysis of open quantum systems beyond the parametric class of OQHOs. This allowed insight to be gained \cite{V_2015b}   on the local sufficiency of linear observers  in the CQF problem for linear quantum plants. This finding suggests that the complicated sets of nonlinear equations for optimal quantum controllers  and filters may appear to be more amenable to solution if they are approached using Hamiltonian structures   similar to those present in the underlying quantum dynamics. Such structures are particularly transparent in closed QHOs. Indeed, these models of linear quantum systems do not involve external bosonic fields and are technically simpler than the above mentioned OQHOs.

We employ this class of models in the present paper and consider a mean square optimal CQF problem for a plant and a directly coupled observer which form a  closed QHO. Since this setting does not use quantum Wiener processes, it simplifies the technical side of the treatment in comparison with \cite{MJ_2012,VP_2013b}. The Hamiltonian  of the plant-observer QHO is a quadratic function of the dynamic variables satisfying the CCRs. When the energy matrix, which specifies the quadratic form of the Hamiltonian, is positive semi-definite, the system variables of the QHO  are either constant or exhibit oscillatory behaviour. This motivates the use of a cost functional (being minimized) in the form of a discounted mean square of an estimation error (with an exponentially decaying weight \cite{B_1965}) with which the observer variables approximate given linear combinations of the plant variables of interest. The performance criterion also involves a quadratic penalty on the plant-observer coupling in order to achieve a compromise between the conflicting requirements of minimizing the estimation error and reducing the back-action of the observer on the plant.  The CQF problem with penalized back-action can also be regarded as a quantum-mechanical counterpart to the classical LQR problem. The use of discounted averages of nonlinear moments of system variables and the presence of optimization makes this setting  different from the time-averaged approach of \cite{P_2014} to CQF in directly coupled  QHOs (see \cite{PH_2015} for a quantum-optical implementation of that approach).

Since discounted moments of system variables for QHOs play an important role throughout the paper, we discuss the computation  of such moments in the state-space and frequency domains for completeness of exposition. Similarly to the variational approach of \cite{VP_2013a,VP_2013b}, we develop first-order necessary conditions of optimality for the CQF problem being considered. These conditions are organized as a set of two algebraic Lyapunov equations (ALEs) for the controllability and observability Gramians which are coupled through another equation for the Hankelian (the product of the Gramians) of the plant-observer composite system. We then employ the Hamiltonian structure of the underlying Heisenberg dynamics and represent this set of equations in terms of the commutators of appropriately transformed  Gramians. This representation allows the third equation to be explicitly solved not only for the plant-observer coupling matrix but also for the energy matrix of an observer of the same dimension as the plant, thus simplifying the set of equations. This reduction is achieved here due to the use of Lie-algebraic techniques (including the Jacobi identity \cite{D_2006}).

The paper is organized as follows.
Section~\ref{sec:QHO} specifies the closed QHOs including its subclass  with positive definite energy matrices. Section~\ref{sec:aver} describes the discounted averaging of moments for system operators in such QHOs both in the time and frequency domains.
Section~\ref{sec:QODE} specifies the direct coupling of quantum plants and coherent quantum observers. Section~\ref{sec:CQF} formulates the discounted mean square optimal CQF problem with penalized back-action. Section~\ref{sec:opt} establishes first-order necessary conditions of optimality for this problem. Section~\ref{sec:lie} represents the optimality conditions in a Lie-algebraic form.
Section~\ref{sec:same} specifies these results to the
case of equally dimensioned plant and observer.
Section~\ref{sec:conc} provides concluding  remarks.

\section{QUANTUM HARMONIC OSCILLATORS}\label{sec:QHO}

Consider a QHO \cite{M_1998}  with dynamic variables $X_1, \ldots, X_n$ (where $n$ is even) which are time-varying self-adjoint operators on a  complex separable Hilbert space $\cH$ satisfying the CCRs
\begin{equation}
\label{XCCR}
    [X(t),X(t)^{\rT}]
    :=
    ([X_j(t),X_k(t)])_{1\< j,k\< n}
    =
    2i
    \Theta,
    \qquad
    X
    :=
    {\begin{bmatrix}
        X_1(t)\\
        \vdots\\
        X_n(t)
    \end{bmatrix}}
\end{equation}
at any instant $t\> 0$ (we will often omit the time arguments for brevity). It is assumed that the \emph{CCR matrix}  $\Theta \in \mA_n$ is nonsingular.  Here, $\mA_n$ denotes the subspace of real antisymmetric matrices  of order $n$. The entries $\theta_{jk}$ of $\Theta$ on the right-hand side of (\ref{XCCR}) represent the scaling operators $\theta_{jk}\cI$, where $\cI$ is the identity operator  on the space $\cH$. The transpose $(\cdot)^{\rT}$ acts on matrices of operators as if the latter were scalars, vectors are organized as columns unless indicated otherwise,  $[\phi,\psi]:= \varphi\psi-\psi\varphi$ is the commutator of operators, and $i:=\sqrt{-1}$ is the imaginary unit. The QHO has a quadratic Hamiltonian
\begin{equation}
\label{HR}
    H := \frac{1}{2}X^{\rT} R X,
\end{equation}
specified by an \emph{energy matrix} $R \in \mS_n$, where $\mS_n$ denotes the subspace of real symmetric matrices of order $n$. Due to (\ref{XCCR}) and (\ref{HR}),  the Heisenberg dynamics of the QHO are governed by a linear ODE
\begin{equation}
\label{Xd}
    \dot{X} = i[H,X] = A X,
\end{equation}
where $A\in \mR^{n\x n}$  is a matrix of constant coefficients  given by
\begin{equation}
\label{A}
    A := 2\Theta R.
\end{equation}
The solution of the ODE (\ref{Xd}) is expressed using the standard matrix exponential as
\begin{equation}
\label{XA}
    X(t)
    =
    \sj_t(X_0)
    :=
    U(t)^{\dagger} X_0 U(t)
    =
    \re^{it \ad_{H_0}}(X_0)
    =
    \re^{tA} X_0,
\end{equation}
where
$\ad_{\alpha}:= [\alpha, \cdot]$, and  the subscript $(\cdot)_0$ indicates the initial values at time $t=0$.
The first three equalities in (\ref{XA}) apply to a general Hamiltonian $H_0$ (that is, not necessarily a quadratic function of $X_0$), and
$  U(t):= \re^{-itH_0}$
is a time-varying unitary operator on $\cH$ (with the adjoint  $U(t)^{\dagger} = \re^{itH_0}$), which specifies the flow $\sj_t$ in (\ref{XA}) acting as a unitary similarity transformation on the system variables. The flow $\sj_t$ preserves the CCRs (\ref{XCCR}) which, in view of the relation
$
    [X(t),X(t)^{\rT}]
    =
    \re^{tA} [X_0, X_0^{\rT}] \re^{tA^{\rT}}
    =
    2i\re^{tA} \Theta \re^{tA^{\rT}}
    =
    2i\Theta
$,
are equivalent to the symplectic property $\re^{tA} \Theta \re^{tA^{\rT}} = \Theta$ of the matrix $\re^{tA}$ for any time $t\> 0$.  The infinitesimal form of this property is
$    A\Theta + \Theta A^{\rT} = 0
$.
This equality corresponds to the PR conditions for OQHOs \cite{JNP_2008,SP_2009} and its fulfillment is ensured by the Hamiltonian structure $A\in \Theta \mS_n$ of the matrix  $A$ in (\ref{A}).

If the energy matrix in (\ref{HR}) is positive semi-definite, $R\succcurlyeq  0$ (and hence,  has a square root $\sqrt{R} \succcurlyeq 0$) then $A=2\Theta \sqrt{R}\sqrt{R}$ is isospectral to the matrix $2\sqrt{R}\Theta \sqrt{R} \in \mA_n$ whose eigenvalues are purely imaginary \cite{HJ_2007}. In the case
 $R\succ 0$, this follows directly from the similarity transformation
\begin{equation}
\label{AR}
    A = R^{-1/2} (2\sqrt{R}\Theta \sqrt{R}) \sqrt{R}
\end{equation}
(see, for example, \cite{P_2014})
which allows $A$ to be diagonalized as
\begin{equation}
\label{AV}
    A = i V \Omega W,
    \qquad
    W := V^{-1},
    \qquad
    \Omega := \diag_{1\< k\< n}(\omega_k).
\end{equation}
Here, $W:= (w_{jk})_{1\< j,k\< n}\in \mC^{n\x n}$ is the inverse of a nonsingular matrix  $V:= (v_{jk})_{1\< j,k\< n}\in \mC^{n\x n}$ whose columns $V_1,\ldots, V_n\in \mC^n$ are the eigenvectors of $A$, and $\Omega:= \diag_{1\< k\< n}(\omega_k)\in \mR^{n \x n}$ is a diagonal matrix of frequencies of the QHO. These frequencies  (which should not be confused with the eigenvalues of the Hamiltonian $H$ as an operator on $\cH$ describing the energy levels of the QHO \cite{S_1994}) are nonzero and symmetric about the origin, and, without loss of generality, are assumed to be arranged so that
\begin{equation}
\label{sym}
    \omega_k = -\omega_{k + \frac{n}{2}} >0,
    \qquad
    k=1, \ldots, \frac{n}{2}.
\end{equation}
Note that  $\sqrt{R} V$ is a unitary matrix  whose columns are the eigenvectors of the matrix $i\sqrt{R}\Theta \sqrt{R} \in \mH_n$ in view of (\ref{AR}); see also the proof of Williamson's symplectic diagonalization  theorem \cite{W_1936,W_1937} in \cite[pp. 244--245]{D_2006}. Here, $\mH_n$ is the subspace of complex Hermitian matrices of order $n$. Substitution of (\ref{AV}) into (\ref{XA}) leads to
\begin{equation}
\label{XAV}
    X(t) = V \re^{it\Omega} W X_0.
\end{equation}
Due to the presence of the matrix $\re^{it\Omega} = \diag_{1\< k\< n}(\re^{i\omega_k t})$ in (\ref{XAV}), the  dynamic variables of the QHO are linear combinations of their initial values whose coefficients are trigonometric polynomials of time:
\begin{equation}
\label{Xt0}
    X_j(t) = \sum_{k,\ell=1}^n c_{jk\ell}\re^{i\omega_k t} X_{\ell}(0),
    \qquad
    j=1, \ldots, n,
\end{equation}
where $c_{jk\ell}:= v_{jk} w_{k\ell}$ are complex parameters which are assembled into rank-one matrices
    $ C_k := (c_{jk\ell})_{1\< j,\ell \< n} = V_kW_k$,
with $W_k$  denoting the $k$th row of $W$. The matrices $C_1, \ldots, C_n$ form a resolution of the identity: $\sum_{k=1}^n C_k  = VW=I_n$. Also,
    $ \overline{C_k} = C_{k + \frac{n}{2}}$
for all $
    k = 1, \ldots, \frac{n}{2}$,
in accordance with (\ref{sym}), whereby (\ref{Xt0}) can be represented in vector-matrix form as
\begin{equation}
\label{XCX}
    X(t)
    =
    \sum_{k=1}^{n/2}
    \big(
    \re^{i\omega_k t} C_k
    +
    \re^{-i\omega_k t} \overline{C_k}
    \big)
    X_0
    =
    2\sum_{k=1}^{n/2}\Re(\re^{i\omega_k t} C_k) X_0,
\end{equation}
where $\overline{(\cdot)}$ denotes the complex conjugate.
Therefore, for any positive integer $d$ and any $d$-index $j:= (j_1, \ldots, j_d)\in \{1, \ldots, n\}^d$, the following degree $d$ monomial of the system variables is also a trigonometric polynomial of time $t$:
\begin{equation}
\label{Xit}
    \Xi_j(t)
    :=
    \rprod_{s=1}^{d}
    X_{j_s}(t)
    =
    \sum_{k,\ell \in \{1, \ldots, n\}^d}\,
    \prod_{s=1}^{d}
    c_{j_sk_s\ell_s}
    \re^{i\omega_{k_s}t}\,
    \Xi_{\ell}(0).
\end{equation}
Here, $\rprod$ denotes the ``rightwards'' ordered product of operators (with the order of multiplication being essential for non-commutative quantum variables),
and the sum is taken over $d$-indices $k:= (k_1, \ldots, k_d), \ell:= (\ell_1, \ldots, \ell_d)\in \{1, \ldots, n\}^d$. Also, note that (\ref{Xt0}) is a particular case of (\ref{Xit}) with $d=1$. The relations (\ref{XAV})--(\ref{Xit}) remain valid in the case $R\succcurlyeq 0$, except that the frequencies $\omega_1, \ldots, \omega_{n/2}$ in (\ref{sym}) are nonnegative.

\section{DISCOUNTED MOMENTS OF SYSTEM OPERATORS}\label{sec:aver}

For any $\tau>0$, we define a linear functional $\bE_{\tau}$ which maps a system operator $\sigma$ of the QHO to the weighted time average
\begin{equation}
\label{bEtau}
    \bE_{\tau}\sigma:= \frac{1}{\tau} \int_0^{+\infty} \re^{-t/\tau} \bE \sigma(t)\rd t.
\end{equation}
Here, $\bE\sigma:= \Tr(\rho \sigma)$  denotes the quantum expectation over the underlying quantum state $\rho$ (which is a positive semi-definite self-adjoint operator on $\cH$ with unit trace).
The weighting function  $\frac{1}{\tau} \re^{-t/\tau}$  in (\ref{bEtau}) is the density of an exponential probability distribution with mean value $\tau$. Therefore, $\tau$ plays the role of an effective horizon for averaging $\bE\sigma$ over time.  This time average (where the relative importance of the quantity of interest decays exponentially)  has the structure of
a discounted cost functional in dynamic programming problems \cite{B_1965}. In particular, if $\bE\sigma(t)$, as a function of time $t\>0$,  is right-continuous at $t=0$, then  $\lim_{\tau\to 0+} \bE_{\tau} \sigma = \bE \sigma_0$. At the other extreme,
the \emph{infinite-horizon average} of $\sigma$ is defined by
\begin{equation}
\label{bEinf}
    \bE_{\infty} \sigma
    :=
    \lim_{\tau\to +\infty} \bE_{\tau}\sigma
    =
    \lim_{\tau\to +\infty}
    \Big(
        \frac{1}{\tau}
        \int_0^\tau
        \bE \sigma(t)
        \rd t
    \Big),
\end{equation}
provided these limits exist. The second of these equalities, whose right-hand side is the Cesaro mean of $\bE \sigma$, follows from the integral version of
the Hardy-Littlewood Tauberian theorem \cite{F_1971}. In particular, (\ref{bEinf}) implies that $|\bE_{\infty}\sigma|\< \limsup_{t\to +\infty}|\bE \sigma(t)|$.

In the case when the QHO has a positive semi-definite energy matrix, the coefficients in (\ref{XCX}) and (\ref{Xit}) are either constant or oscillatory, which makes the time averages (\ref{bEtau}) and (\ref{bEinf}) well-defined for nonlinear functions of the system variables and their moments for any $\tau>0$.
To this end, we will use the characteristic function $\chi_{\tau}: \mR\to \mC$ of the exponential distribution and its pointwise convergence:
\begin{equation}
\label{chi}
    \chi_{\tau}(u)
     :=  \frac{1}{\tau}
    \int_0^{+\infty}
    \re^{-t/\tau}
    \re^{iut}
    \rd t
     =
    \frac{1}{1-iu \tau}
      \to
    \delta_{u0}
     =
    \left\{
        \begin{matrix}
            1 & {\rm if}\  u =0\\
            0 & {\rm if}\  u \ne 0
        \end{matrix}
    \right.,
    \qquad
    {\rm as}\
    \tau\to +\infty,
\end{equation}
where $\delta_{pq}$ is the Kronecker delta.
A combination of (\ref{Xit}) with (\ref{chi}) implies that if the initial system variables of the QHO have finite mixed moments $\bE \Xi_{\ell}(0)$ of order $d$ for all $\ell \in \{1,\ldots, n\}^d$, then such moments have the following time-averaged values (\ref{bEtau}):
\begin{equation}
\label{EXi}
    \bE_{\tau} \Xi_j
    :=
    \frac{1}{\tau}
    \int_0^{+\infty}
    \re^{-t/\tau}
    \bE \Xi_j(t)
    \rd t
    =
    \sum_{k\in \{1, \ldots, n\}^d}
    \chi_{\tau}\Big(\sum_{s=1}^d \omega_{k_s}\Big)\!\!
    \sum_{\ell \in \{1, \ldots, n\}^d}\,
    \prod_{s=1}^{d}
    c_{j_sk_s\ell_s}
    \bE \Xi_{\ell}(0)
\end{equation}
for any $j\in \{1, \ldots, n\}^d$. Hence, the corresponding infinite-horizon average (\ref{bEinf}) takes the form
\begin{equation}
\label{EXiinf}
    \bE_{\infty} \Xi_j
    =
    \sum_{k\in \cK_d}\,
    \sum_{\ell \in \{1, \ldots, n\}^d}\,
    \prod_{s=1}^{d}
    c_{j_sk_s\ell_s}
    \bE \Xi_{\ell}(0),
\end{equation}
where
$
    \cK_d
    :=
    \big\{
        (k_1, \ldots, k_d)\in \{1, \ldots, n\}^d:\
        \sum_{s=1}^{d}\omega_{k_s} = 0
    \big\}
$
is a subset of $d$-indices associated with the frequencies $\omega_1, \ldots, \omega_n$ of the QHO from (\ref{AV}).
For every even $d$, the set $\cK_d$ is nonempty due to the central symmetry of the frequencies.

The linear functional $\bE_{\tau}$ in (\ref{EXi}) and its limit $\bE_{\infty}$ in (\ref{EXiinf}) are extendable to polynomials and more general functions $\sigma:= f(X)$ of the system variables, provided $X_0$ satisfies appropriate integrability conditions. Such an extension of $\bE_{\infty}$, which involves the Cesaro mean, is similar to the argument used in the context of Besicovitch spaces of almost periodic functions \cite{B_1954}.
If the system is in an invariant state $\rho$ (which, therefore, satisfies $[H_0, \rho] = 0$), then the quantum expectation $\bE \sigma = \Tr( \rho \re^{it\ad_{H_0}}(\sigma_0)) = \Tr(\re^{-it\ad_{H_0}}(\rho)\sigma_0) = \Tr(\rho\sigma_0) $ is time-independent for any system operator $\sigma_0$ evolved by the flow (\ref{XA}). In this case, the time averaging in (\ref{bEtau}) becomes redundant.
However, the subsequent discussion is concerned with general (not necessarily invariant) quantum states $\rho$.

Of particular use for our  purposes is the following state-space computation of the discounted time average (\ref{bEtau}) for second moments of the system variables, which is concerned with finite values of $\tau$ and does not employ
the imaginarity of the spectrum of $A$.
To this end, we note that $\bE(XX^{\rT})\in \mH_n^+$  at every moment of time due to
the
generalized Heisenberg uncertainty principle \cite{H_2001}, where $\mH_n^+$ denotes the set of complex positive semi-definite Hermitian matrices of order $n$. Furthermore, $\Im \bE(XX^{\rT}) = \Theta$ remains unchanged in view of the preservation of the CCRs (\ref{XCCR}) mentioned above. Also, with any  Hurwitz matrix $\alpha$, we associate a linear operator $\bL(\alpha, \cdot)$
which maps an appropriately dimensioned matrix $\beta$ to a unique solution $\gamma=\bL(\alpha,\beta)$  of the
ALE $\alpha \gamma +\gamma\alpha^{\rT}+\beta= 0$:
\begin{equation}
\label{ILO}
  \bL(\alpha,\beta):= \int_0^{+\infty} \re^{t\alpha}\beta \re^{t\alpha^{\rT}}\rd t.
\end{equation}

\begin{lem}
\label{lem:PALE}
Suppose the initial dynamic variables of the QHO have finite second moments (that is, $\bE(X_0^{\rT}X_0)<+\infty$) whose real parts form the matrix
\begin{equation}
\label{Sigma}
  \Sigma:= \Re \bE(X_0X_0^{\rT}).
\end{equation}
Also, let the effective time horizon $\tau>0$ be bounded above as
\begin{equation}
\label{taumax}
    \tau < \frac{1}{2\max(0,\, \ln \br(\re^A))},
\end{equation}
where $\br(\cdot)$ denotes the spectral radius of a matrix.
Then the matrix of the real parts of the discounted second moments of the dynamic variables can be computed as
\begin{equation}
\label{P}
    P:= \Re \bE_{\tau}(XX^{\rT}) = \frac{1}{\tau}\bL(A_{\tau},\Sigma)
\end{equation}
through the operator (\ref{ILO}). That is, $P$ is a unique solution of the following ALE with a Hurwitz matrix $A_{\tau}$:
\begin{equation}
\label{PALE}
    A_{\tau}
    P
    +
    P
    A_{\tau}^{\rT} + \frac{1}{\tau}\Sigma= 0,
    \qquad
    A_{\tau}:= A - \frac{1}{2\tau}I_n.
\end{equation}
\end{lem}
\begin{proof}
By combining (\ref{XA}) with (\ref{Sigma}), it follows that $\Re \bE(X(t)X(t)^{\rT}) = \re^{tA} \Sigma \re^{tA^{\rT}}$ for any $t\>0$. Hence, in application to (\ref{P}),  the time average (\ref{bEtau}) can be computed  as
$$
    P
     =
    \frac{1}{\tau}
    \int_0^{+\infty}
    \re^{-t/\tau}
    \Re \bE(X(t)X(t)^{\rT})
    \rd t
      =
    \frac{1}{\tau}
    \int_0^{+\infty}
    \re^{-t/\tau}
    \re^{tA} \Sigma \re^{tA^{\rT}}
    \rd t
      =
    \frac{1}{\tau}
    \int_0^{+\infty}
    \re^{tA_{\tau}} \Sigma \re^{tA_{\tau}^{\rT}}
    \rd t
    =\frac{1}{\tau}\bL(A_{\tau},\Sigma),
$$
thus establishing the representation (\ref{P}).
Here, the matrix $A_{\tau}$, given by (\ref{PALE}),
is Hurwitz due to the condition (\ref{taumax}).
\end{proof}

In view of (\ref{PALE}), the matrix $P$ is the controllability Gramian \cite{KS_1972} of the pair $(A_{\tau}, \sqrt{\tau^{-1}\Sigma})$.
In contrast to similar ALEs for steady-state covariance matrices in dissipative OQHOs \cite{EB_2005} (where the corresponding matrix $A$ itself is Hurwitz), the term   $\frac{1}{\tau}\Sigma$ in (\ref{PALE}) comes from the initial condition (\ref{Sigma}) instead of the Ito matrix of the quantum Wiener process \cite{H_2001,HP_1984,P_1992}.
Since $A$ is a Hamiltonian matrix (and hence, its spectrum is symmetric about the imaginary axis), the condition (\ref{taumax}) is equivalent to the eigenvalues  of $A$ being contained in the strip $\big\{z \in \mC:\ |\Re z|< \frac{1}{2\tau}\big\}$.
For any $\tau>0$ satisfying (\ref{taumax}), a frequency-domain representation of the matrix $P$ in (\ref{P}) is
\begin{equation}
\label{PLap}
    P
      =
    \frac{1}{2\pi\tau}
    \Re
    \int_{-\infty}^{+\infty}
    F\Big(\frac{1}{2\tau} + i\omega\Big)
    \Gamma
    F
    \Big(\frac{1}{2\tau} + i\omega\Big)^*
    \rd \omega
     =
    \frac{1}{2\pi\tau}
    \Im
    \int_{\Re s=\frac{1}{2\tau}}
    F(s)
    \Gamma
    F(s)^*
    \rd s,
\end{equation}
where $(\cdot)^*:= (\overline{(\cdot)})^{\rT}$ denotes the complex conjugate transpose.
Here,
 $   \Gamma:= \bE(X_0 X_0^{\rT}) = \Sigma +i\Theta$
is the matrix of second moments of the initial system variables, and
 $   F(s):= (sI_n-A)^{-1}$
is the transfer function (where the complex variable $s$ satisfies     $\Re s > \ln\br(\re^A)$) which relates
the Laplace transform
 $   \wt{X}(s)
    :=
    \int_0^{+\infty}
    \re^{-st} X(t)\rd t$
of the quantum process $X$ from (\ref{XA}) to its initial value $X_0$ as
$
    \wt{X}(s)
    =
    \int_0^{+\infty}
    \re^{-t(sI_n-A)}\rd tX_0  = F(s) X_0
$.
The representation (\ref{PLap}) is obtained by applying an operator version of the Plancherel theorem to the inverse Fourier transform
$
    \re^{-\frac{t}{2\tau}} X(t)
    =
    \frac{1}{2\pi}
    \int_{-\infty}^{+\infty}
    \re^{i\omega t}
    \wt{X}\big(\frac{1}{2\tau} + i\omega\big)
    \rd \omega$ for
$
    t\> 0
$
under the condition (\ref{taumax}).

\section{DIRECTLY COUPLED  QUANTUM PLANT AND COHERENT QUANTUM OBSERVER}\label{sec:QODE}

Consider a direct coupling of a quantum plant and a coherent quantum observer which form a closed QHO whose Hamiltonian $H$ is given by
\begin{equation}
\label{H}
    H
    :=
    \frac{1}{2}\cX^{\rT} R \cX,
    \qquad
    \cX
    :=
    {\begin{bmatrix}
        X\\
        \xi
    \end{bmatrix}},
    \qquad
    X
    :=
    {\begin{bmatrix}
    X_1\\
    \vdots\\
    X_n
    \end{bmatrix}},
    \qquad
    \xi
    :=
    {\begin{bmatrix}
    \xi_1\\
    \vdots\\
    \xi_{\nu}
    \end{bmatrix}},
\end{equation}
where $R \in \mS_{n+\nu}$ is the plant-observer energy matrix.
Here, $X_1, \ldots, X_n$ and $\xi_1, \ldots, \xi_{\nu}$ are the dynamic variables of the plant and the observer, respectively, with both dimensions $n$ and $\nu$ being  even.  The plant and observer variables are time-varying self-adjoint operators on the tensor-product space $\cH:= \cH_1\ox \cH_2$, where   $\cH_1$ and $\cH_2$ are initial complex separable Hilbert  spaces of the plant and the observer (which can be copies of a common Hilbert space). These quantum variables are assumed to satisfy the CCRs with a block-diagonal CCR matrix $\Theta$:
\begin{equation}
\label{Theta}
    [\cX,\cX^{\rT}]
    =
    2i \Theta,
    \qquad
    \Theta
    :=
    \diag_{k=1,2}(\Theta_k),
\end{equation}
where $\Theta_1\in \mA_n$ and $\Theta_2 \in \mA_{\nu}$ are nonsingular CCR matrices of the plant and the observer, respectively. For what follows, the plant-observer energy matrix $R$ in (\ref{H}) is partitioned as
\begin{equation}
\label{R}
    R:=
    {\begin{bmatrix}
        K  & L\\
        L^{\rT} & M
    \end{bmatrix}}.
\end{equation}
Here, $K\in \mS_n$ and $M\in \mS_{\nu}$ are the energy matrices of the plant and the observer which specify their free Hamiltonians $
    H_1  := \frac{1}{2}X^{\rT}K X
$ and $
    H_2:= \frac{1}{2}\xi^{\rT}M \xi$. Also, $L\in \mR^{n\x \nu}$ is the \emph{plant-observer coupling matrix} which parameterizes the interaction Hamiltonian $
    H_{12}:= \frac{1}{2} (X^{\rT}L\xi + \xi^{\rT}L^{\rT}X) =
    \Re(X^{\rT}L\xi)
$,
where $\Re(\cdot)$ applies to operators (and matrices of operators) so that $\Re N := \frac{1}{2}(N+N^{\#})$ consists of self-adjoint operators. Accordingly, the total Hamiltonian $H$ in (\ref{H}) is representable as $H = H_1 + H_2 + H_{12}$.
In view of (\ref{H})--(\ref{R}),  the Heisenberg dynamics of the composite system are governed by a linear ODE
\begin{equation}
\label{cXdot}
    \dot{\cX} = i[H,\cX] =\cA \cX.
\end{equation}
Here, in accordance with the partitioning of $\cX$ in (\ref{H}), the matrix $\cA\in \mR^{(n+\nu)\x (n+\nu)}$ is split into appropriately dimensioned blocks as
\begin{equation}
\label{cA}
    \cA
    :=
    2\Theta R
    =
    2
    {\begin{bmatrix}
        \Theta_1 K & \Theta_1 L\\
        \Theta_2 L^{\rT} & \Theta_2 M
    \end{bmatrix}}
    =
{\begin{bmatrix} A & BL\\ \beta L^{\rT} & \alpha\end{bmatrix}},
\end{equation}
with the ODE (\ref{cXdot}) being representable as a set of two ODEs
\begin{align}
\label{Xdot}
    \dot{X}  & = AX +  B \eta,\\
\label{xidot}
    \dot{\xi}  & = \alpha \xi + \beta Y,
\end{align}
where         $A:= 2\Theta_1 K$,
        $B := 2\Theta_1$, $
        \alpha := 2\Theta_2 M$,
        $\beta := 2\Theta_2$, and
\begin{equation}
\label{Yeta}
    Y  := L^{\rT} X,
    \qquad
    \eta := L\xi.
\end{equation}
The vector $\eta$ drives the plant variables in (\ref{Xdot}), thus resembling a classical actuator signal. The observer variables in (\ref{xidot}) are driven by the plant variables through the vector $Y$ which  corresponds to a classical  observation output from the plant.  However, the quantum mechanical nature of $Y$ and $\eta$ (which consist of time-varying self-adjoint operators on $\cH$) makes them qualitatively different from the classical signals \cite{AM_1989,KS_1972}. In particular, since the plant and the observer being considered form a fully quantum system which does not involve measurements,
$Y$ is not an observation signal in the usual control theoretic sense. In order to emphasize this distinction, such observers are referred to as coherent (that is, measurement-free) quantum observers \cite{JNP_2008,L_2000,MJ_2012,NJP_2009,VP_2013b,WM_1994_paper}.
In addition to the noncommutativity of the dynamic variables, specified by the CCRs (\ref{Theta}), the quantum mechanical nature of the setting manifests itself in the fact that the ``observation'' and ``actuation'' channels  in (\ref{Yeta}) depend on the same matrix $L$. This coupling  between the ODEs (\ref{Xdot}), (\ref{xidot}) is closely related to the Hamiltonian structure $\cA\in \Theta \mS_{n+\nu} $ of the matrix $\cA$ in (\ref{cA}).
Therefore, the ``quantum information flow'' from the plant to the observer through $Y$  has a ``back-action'' effect on the plant dynamics through $\eta$.

Assuming that the plant energy matrix $K$ is
fixed,
the matrices $L$ and $M$ can be varied so as to achieve desired properties for the plant-observer QHO under constraints on the plant-observer coupling. To this end, for a given effective time horizon $\tau>0$, the observer will be called \emph{$\tau$-admissible} if the matrix $\cA$ in (\ref{cA}) satisfies
\begin{equation}
\label{taugood}
    \tau < \frac{1}{2\max(0,\, \ln \br(\re^{\cA}))},
\end{equation}
cf. (\ref{taumax}) of Lemma~\ref{lem:PALE}. The corresponding pairs $(L,M)$ form an open subset of $\mR^{n\x \nu}\x \mS_{\nu}$ which depends on $\tau$. In application  to the plant-observer system,  the discussions of Section~\ref{sec:QHO} show that
if the matrix $R$ in (\ref{R}) is positive semi-definite (and hence, $\cA$ has a purely imaginary spectrum), then such an observer is $\tau$-admissible for any $\tau>0$.
In this case, any system operator (with appropriate finite moments) in the plant-observer QHO lends itself to the discounted averaging, described in Section~\ref{sec:aver}, for any $\tau>0$.
For what follows, it is assumed that the initial plant and observer variables have a block-diagonal matrix of second moments:
\begin{equation}
\label{Sig}
  \Sigma
  :=
  \Re \bE(\cX_0\cX_0^{\rT})
  =
  \diag_{k=1,2}(\Sigma_k),
\end{equation}
where $\Sigma_k+i\Theta_k \succcurlyeq 0$. In the zero mean case $\bE \cX_0 = 0$, this corresponds to $X_0$ and $\xi_0$ being uncorrelated. A physical rationale for the absence of initial correlation is that the observer is prepared independently of the plant and then brought into interaction at time $t=0$.
If the plant and the observer remained uncoupled  (which would  correspond to $L=0$), then, in view of Lemma~\ref{lem:PALE} and (\ref{Sig}), their variables would remain uncorrelated (in the sense that $\bE (X\xi^{\rT}) = 0$).
In the general case of plant-observer coupling $L\ne 0$,  the matrix
\begin{equation}
\label{cP}
    \cP
    :=
    {\begin{bmatrix}
        \cP_{11} & \cP_{12}\\
        \cP_{21} & \cP_{22}
    \end{bmatrix}}
    :=
\Re \bE_{\tau} (\cX\cX^{\rT}),
\end{equation}
which is split into blocks similarly to $\cA$ in (\ref{cA}),
coincides with the controllability Gramian of the pair $(\cA_{\tau}, \sqrt{\tau^{-1}\Sigma})$ and satisfies an appropriate ALE:
\begin{equation}
\label{cPALE}
    \cP = \frac{1}{\tau}\bL(\cA_{\tau}, \Sigma),
\qquad
    \cA_{\tau}
    :=
    \cA - \frac{1}{2\tau}I_{n+\nu},
\end{equation}
provided the observer is $\tau$-admissible in the sense of (\ref{taugood}), thus making the matrix $\cA_{\tau}$ Hurwitz.  Here, $\Sigma$ is the initial covariance condition from (\ref{Sig}).

\section{DISCOUNTED MEAN SQUARE OPTIMAL COHERENT QUANTUM FILTERING PROBLEM}\label{sec:CQF}

If the plant energy matrix satisfies $K\succcurlyeq  0$, then the set of $\tau$-admissible observers is nonempty for any $\tau>0$. This set contains observers with $M\succcurlyeq 0$ and $L:=\sqrt{K} \Lambda \sqrt{M}$, where $\Lambda\in \mR^{n\x \nu}$ is an arbitrary matrix whose largest singular value satisfies $\|\Lambda\|\< 1$. Indeed, for any such observer, $R=\diag(\sqrt{K},\sqrt{M}){\small \begin{bmatrix}I_n & \Lambda\\ \Lambda^{\rT} & I_{\nu}\end{bmatrix}}\diag(\sqrt{K},\sqrt{M})\succcurlyeq 0$, and hence, the matrix $\cA$ in (\ref{cA}) has a purely imaginary spectrum. Consider a CQF problem of minimizing  a discounted mean square cost functional $\cZ$ over the plant-observer coupling matrix $L$  and the observer energy matrix $M$ subject to the constraint (\ref{taugood}):
\begin{equation}
\label{cZ}
    \cZ
    :=
    \bE_{\tau} Z
        \longrightarrow \min,
\end{equation}
where $\tau>0$ is a given effective time horizon for the discounted averaging (\ref{bEtau}). This averaging is applied to the process
\begin{equation}
    \label{ZE}
    Z:= E^{\rT}E + \lambda\eta^{\rT} \Pi \eta = \cX^{\rT}\cC^{\rT}\cC\cX
\end{equation}
which is a time-varying self-adjoint operator on the plant-observer space $\cH$ defined in terms of
the vectors $\cX$, $\eta$ from  (\ref{H}), (\ref{Yeta}), with
\begin{equation}
\label{ES}
    E:= S_1X-S_2\xi,
    \qquad
    \cC:= {\begin{bmatrix} S_1 & -S_2\\ 0 & \sqrt{\lambda \Pi} L\end{bmatrix}}.
\end{equation}
Also,
$S_1\in \mR^{p\x n}$, $S_2\in \mR^{p\x \nu}$ and $\Pi \in \mS_n$ are given matrices, with $\Pi\succ 0$,  which, together with a scalar parameter $\lambda>0$, determine the matrix $\cC\in \mR^{(p+n)\x (n+\nu)}$ and its dependence on the coupling matrix $L$.
 The matrix  $S_1$ specifies linear combinations of the plant variables of interest to be approximated by given linear functions of the observer variables specified by the matrix $S_2$. Accordingly, the $p$-dimensional  vector $E$ in (\ref{ES}) (which consists of time-varying self-adjoint operators on $\cH$)
is interpreted as an \emph{estimation error}. In addition to the discounted mean square $\bE_{\tau} (E^{\rT} E)$ of the estimation error,  the cost functional $\cZ$ in (\ref{cZ}) involves a quadratic penalty $\bE_{\tau}(\eta^{\rT} \Pi\eta)$ for the observer back-action on the covariance dynamics of the plant,
with $\lambda$ being the relative weight of this penalty in $\cZ$. 

The parameter $\lambda$ in the CQF problem (\ref{cZ})--(\ref{ES}) quantifies a compromise  between the conflicting requirements of minimizing the estimation error and reducing the back-action.
In fact, $\cZ$ is organized as the Lagrange function for a related problem of minimizing the discounted mean square of the estimation error subject to a  weighted mean square constraint on the plant-observer coupling:
\begin{equation}
\label{CQF2}
    \bE_{\tau} (E^{\rT} E)
    \longrightarrow
    \min,
    \qquad
    \bE_{\tau} (\eta^{\rT}\Pi\eta)\< r.
\end{equation}
In this formulation,  $\lambda$ plays the role of a Lagrange multiplier which is found so as to make the solution of (\ref{cZ}) saturate the constraint in (\ref{CQF2}) for a given threshold $r> 0$.

In the particular case  of $S_2=0$, the CQF problem  (\ref{cZ})--(\ref{ES}) can be regarded as a quantum mechanical analogue of the LQR problem \cite{AM_1989,KS_1972} in view of the analogy between the observer output $\eta$ and classical actuation signals discussed in Section~\ref{sec:QODE}.  The presence of the  quantum expectation  of a nonlinear function of system variables in (\ref{cZ}) and the optimization requirement make this setting different from the time-averaged approach of \cite{P_2014,PH_2015}.

\section{FIRST-ORDER NECESSARY CONDITIONS OF OPTIMALITY}\label{sec:opt}

The following theorem, which  provides first-order necessary conditions of optimality for the CQF problem (\ref{cZ})--(\ref{ES}),   employs
the Hankelian
\begin{equation}
\label{cE}
    \cE
    :=
    {\begin{bmatrix}
        \cE_{11} & \cE_{12}\\
        \cE_{21} & \cE_{22}
    \end{bmatrix}}
    :=
    \cQ\cP.
\end{equation}
This matrix is associated with the controllability Gramian $\cP$ from (\ref{cP}), (\ref{cPALE}) and the observability Gramian $\cQ$ of $(\cA_{\tau}, \cC)$ which is a unique solution of the corresponding ALE:
\begin{align}
\label{cQALE}
    \cQ :=
        {\begin{bmatrix}
        \cQ_{11} & \cQ_{12}\\
        \cQ_{21} & \cQ_{22}
    \end{bmatrix}}
    =
    \bL(\cA_{\tau}^{\rT}, \cC^{\rT}\cC).
\end{align}
 The matrices $\cE$ and $\cQ$ are partitioned into appropriately dimensioned blocks $(\cdot)_{jk}$ similarly to the matrix $\cP$.

\begin{thm}
\label{th:stat}
Suppose the plant energy matrix satisfies $K\succcurlyeq 0$, and the directly coupled observer is $\tau$-admissible in the sense of (\ref{taugood}). Then the observer is a stationary point of the CQF problem (\ref{cZ})--(\ref{ES}) if and only if the Hankelian $\cE$ in (\ref{cE}) 
satisfies
\begin{align}
\label{dcZdL0}
    \Theta_1 \cE_{12}
    -\cE_{21}^{\rT}\Theta_2
    & =\frac{\lambda}{2} \Pi L \cP_{22},\\
\label{dcZdM0}
    \Theta_2 \cE_{22}-\cE_{22}^{\rT}\Theta_2& =0.
\end{align}
\end{thm}
\begin{proof}
By substituting (\ref{ZE}) into (\ref{cZ}), the cost functional is expressed in terms of the matrix $\cP$ from (\ref{cP}) as
\begin{equation}\label{ZP}
  \cZ = \bra \cC^{\rT}\cC, \bE_{\tau}(\cX\cX^{\rT})\ket = \bra \cC^{\rT}\cC, \cP\ket,
\end{equation}
where $\bra\cdot, \cdot\ket$ denotes the Frobenius inner product of matrices.
By using (\ref{cPALE}) and the adjoint $\bL(\cA_{\tau}, \cdot)^{\dagger} = \bL(\cA_{\tau}^{\rT}, \cdot)$ of the operator (\ref{ILO}), the cost $\cZ$ in (\ref{ZP}) is represented in terms of the observability Gramian $\cQ$ from (\ref{cQALE}) as
\begin{equation}
\label{ZQ}
    \cZ
    =
    \frac{1}{\tau}
    \Bra
        \cC^{\rT}\cC,
        \bL(\cA_{\tau}, \Sigma)
    \Ket
    =
    \frac{1}{\tau}
    \Bra
        \bL(\cA_{\tau}^{\rT}, \cC^{\rT}\cC),
        \Sigma
    \Ket
    =
    \frac{1}{\tau}
    \bra \cQ, \Sigma\ket.
\end{equation}
Since the matrix $\cA_{\tau}$ in (\ref{cPALE}) is Hurwitz due to the $\tau$-ad\-missibi\-li\-ty constraint (\ref{taugood}), the representation (\ref{ZQ}) shows that $\cZ$ inherits  a smooth dependence on  $L$ and $M$ from $\cQ$. The composite function $(L,M)\mapsto (\cA, \cC)\mapsto \cQ$ has the first variation
\begin{equation}
\label{dcQ}
    \delta \cQ = \bL(\cA_{\tau}^{\rT}, \,(\delta \cA)^{\rT} \cQ + \cQ \delta \cA + (\delta \cC)^{\rT}\cC + \cC^{\rT} \delta \cC),\!
\end{equation}
where use is made of the ALE in (\ref{cQALE}), and  the first variations of the matrices $\cA$ in (\ref{cA}) and $\cC$ in (\ref{ES}) with respect to $L$ and $M$ are
\begin{equation}
\label{dcAdcC}
    \delta\cA = 2\Theta {\begin{bmatrix} 0 & \delta L \\ \delta L^{\rT} & \delta M\end{bmatrix}},
    \qquad
    \delta\cC = {\begin{bmatrix} 0 & 0\\ 0 & \sqrt{\lambda \Pi}\delta L\end{bmatrix}}.
\end{equation}
The first variation of $\cZ$ in (\ref{ZQ}) can now be computed by combining the duality argument above with (\ref{dcQ}) and (\ref{dcAdcC}) as
\begin{align}
\nonumber
\delta\cZ
    =&
    \frac{1}{\tau}
    \Bra
        \bL(\cA_{\tau}^{\rT}, (\delta \cA)^{\rT} \cQ + \cQ \delta \cA + (\delta \cC)^{\rT}\cC + \cC^{\rT} \delta \cC),
        \Sigma
    \Ket\\
\nonumber
    =&
    \Bra
        (\delta \cA)^{\rT} \cQ + \cQ \delta \cA + (\delta \cC)^{\rT}\cC + \cC^{\rT} \delta \cC,
        \cP
    \Ket\\
\nonumber
    =&
    2\Bra
        \cE, \delta \cA
    \Ket
    +
    2
    \Bra
        \cC\cP, \delta \cC
    \Ket    \\
\nonumber
    =&
    -
    4
    \Bra
     \Theta \cE,
     { \begin{bmatrix} 0 & \delta L \\ \delta L^{\rT} & \delta M\end{bmatrix}}
    \Ket
    +
    2 \Bra \cC\cP, { \begin{bmatrix} 0 & 0\\ 0 & \sqrt{\lambda \Pi} \delta L\end{bmatrix}}\Ket\\
\nonumber
    =&
    -
    4
    \Bra
     \bS(\Theta \cE),
     { \begin{bmatrix} 0 & \delta L \\ \delta L^{\rT} & \delta M\end{bmatrix}}
    \Ket
    +
    2 \Bra (\cC\cP)_{22}, \sqrt{\lambda \Pi}\delta L\Ket\\
\nonumber
    =&
    -8\Bra
        \bS(\Theta \cE)_{12},
        \delta L
    \Ket
    -4\Bra
        \bS(\Theta \cE)_{22},
        \delta M
    \Ket\\
\nonumber
    & +
2 \bra \sqrt{\lambda \Pi} L\cP_{22}, \sqrt{\lambda \Pi}\delta L\ket\\
\label{delcZ}
    =&
2 \Bra \lambda \Pi L\cP_{22} - 4\bS(\Theta \cE)_{12}, \delta L\Ket
     -
    4
    \Bra
     \bS(\Theta \cE)_{22},
     \delta M
    \Ket
\end{align}
(see, for example, \cite{VP_2013a} for similar calculations). Here,  $\bS(N):= \frac{1}{2}(N+N^{\rT})$ denotes the symmetrizer of matrices, so that
\begin{equation}
\label{bSTE}
    \bS(\Theta \cE)
      =
    \frac{1}{2}
    (\Theta \cE - \cE^{\rT}\Theta)
      =
    \frac{1}{2}
    {\begin{bmatrix}
        \Theta_1 \cE_{11}-\cE_{11}^{\rT}\Theta_1  &
        \Theta_1 \cE_{12}-\cE_{21}^{\rT}\Theta_2\\
        \Theta_2 \cE_{21}-\cE_{12}^{\rT}\Theta_1 &
        \Theta_2 \cE_{22}-\cE_{22}^{\rT}\Theta_2\end{bmatrix}}.
\end{equation}
A combination of (\ref{delcZ}) with (\ref{bSTE}) leads to the  partial Frechet derivatives of $\cZ$ on the Hilbert spaces $\mR^{n\x \nu}$ and $\mS_{\nu}$:
\begin{align}
\label{dcZdL}
    \d_L \cZ
    & =
    2(\lambda \Pi L\cP_{22} - 4\bS(\Theta \cE)_{12})
      = 2 (\lambda \Pi L\cP_{22} -
    2
    (\Theta_1 \cE_{12}-\cE_{21}^{\rT}\Theta_2)),\\
\label{dcZdM}
    \d_M \cZ
    & = -4\bS(\Theta_2 \cE_{22})
     = -2(\Theta_2 \cE_{22}-\cE_{22}^{\rT}\Theta_2).
\end{align}
By equating the Frechet derivatives (\ref{dcZdL}) and (\ref{dcZdM}) to zero, the stationarity of $\cZ$ with respect to $L$ and $M$ is equivalent to (\ref{dcZdL0}) and (\ref{dcZdM0}). \end{proof}

The relation (\ref{bSTE}) implies that  the fulfillment of the first-order optimality conditions (\ref{dcZdL0}) and (\ref{dcZdM0}) for the observer is equivalent to the existence of a matrix $N\in \mS_n$ such that
\begin{equation}
\label{STE}
    \Theta \cE - \cE^{\rT} \Theta
    =
    \frac{1}{2}
    {\begin{bmatrix}
    N & \lambda \Pi L\cP_{22}\\
    \lambda \cP_{22}L^{\rT}\Pi & 0
    \end{bmatrix}}.
\end{equation}
Here, the zero block corresponds to (\ref{dcZdM0}) whereby the matrix $\cE_{22}$ is skew-Hamiltonian in the sense that $\cE_{22} \in \Theta_2^{-1} \mA_{\nu}$. If $\cP_{22}\succ 0$, then, in view of the assumption $\Pi\succ 0$,  (\ref{dcZdL0}) implies that the optimal coupling matrix can be represented as
\begin{equation}
\label{L}
    L
    =
    \frac{2}{\lambda}
    \Pi^{-1}
    (\Theta_1 \cE_{12}
    -\cE_{21}^{\rT}\Theta_2)
    \cP_{22}^{-1}.
\end{equation}
Together with the ALEs (\ref{cPALE}), (\ref{cQALE}) and the relation (\ref{L}),
the optimality condition (\ref{dcZdM0}) for the observer energy matrix $M$ also admits a more explicit form. This step is less straightforward and is  considered in the next two sections.

\section{LIE-ALGEBRAIC REPRESENTATION OF THE OPTIMALITY CONDITIONS}\label{sec:lie}

Associated with the Gramians $\cP$ and $\cQ$ from (\ref{cP}) and (\ref{cQALE}) are the matrices
\begin{equation}
\label{PQ}
    P:= \cP \Theta^{-1},
    \qquad
    Q:= \Theta \cQ
\end{equation}
which belong to the same subspace $\Theta \mS_{n+\nu}$ of Hamiltonian matrices as $\cA$ in (\ref{cA}). The property $P \in \Theta \mS_{n+\nu}$ follows from $\Theta^{-1}\cP \Theta^{-1}\in \mS_{n+\nu}$.  The linear space $\Theta \mS_{n+\nu}$, equipped with the commutator $[\cdot, \cdot]$ (as a non-associative  antisymmetric multiplication satisfying the Jacobi identity), forms a Lie algebra \cite{D_2006}.

\begin{lem}
\label{lem:lie}
The ALEs (\ref{cPALE}), (\ref{cQALE}) and the optimality conditions (\ref{dcZdL0}), (\ref{dcZdM0}) for the CQF problem (\ref{cZ})--(\ref{ES}) admit a Lie-algebraic form in terms of the matrices $P$, $Q$ from (\ref{PQ}):
\begin{align}
\label{cPQALElie1_cPQALElie2}
    [\cA, P] &= \frac{1}{\tau}(P-\Sigma\Theta^{-1}),
    \qquad
    [\cA, Q]  = \Theta \cC^{\rT}\cC - \frac{1}{\tau} Q,\\
\label{optlie1_optlie2}
    [Q,P]_{12}
    & =
    \frac{\lambda}{2}
    \Pi LP_{22},
    \qquad\quad\ \
    [Q,P]_{22}
     =
    0,
\end{align}
where $[Q,P]_{12}$ and $[Q,P]_{22}$ are the corresponding blocks of the Hamiltonian matrix $[Q,P]\in \Theta  \mS_{n+\nu}$.
\end{lem}
\begin{proof}
The Hamiltonian structure of the matrix $\cA$  in (\ref{cA}) implies that $    \cA^{\rT} = -\Theta^{-1}\cA \Theta $. Hence, (\ref{cPALE}) and (\ref{PQ}) imply that
$$
    \cA_{\tau}\cP + \cP \cA_{\tau}^{\rT} = \Big( [\cA, P] - \frac{1}{\tau}P \Big)\Theta,
     \qquad
     \cA_{\tau}^{\rT}\cQ   + \cQ \cA_{\tau} = -\Theta^{-1}\Big([\cA, Q] + \frac{1}{\tau}Q\Big) .
$$
These relations lead to the Lie-algebraic representations (\ref{cPQALElie1_cPQALElie2}) for the ALEs (\ref{cPALE}), (\ref{cQALE}).
The symmetry of the Gramians $\cP$, $\cQ$ and a combination of (\ref{cE}) with (\ref{PQ}) imply that
$$
    \Theta \cE - \cE^{\rT} \Theta
     =
    (\Theta \cQ\cP\Theta^{-1} - \cP\Theta^{-1}\Theta\cQ) \Theta
      =
    [Q,P]\Theta
$$
for any $\tau$-admissible observer. By substituting this identity into (\ref{STE}) and using the relation $\cP_{22}\Theta_2^{-1} = P_{22}$, the optimality conditions (\ref{dcZdL0}) and (\ref{dcZdM0}) admit the Lie-algebraic representations (\ref{optlie1_optlie2}).
\end{proof}

In view of (\ref{cPQALElie1_cPQALElie2}),
  $$  P     = (\cI - \tau \ad_{\cA})^{-1}(\Sigma\Theta^{-1}),
  \qquad
    Q     = \tau (\cI + \tau \ad_{\cA})^{-1}(\Theta \cC^{\rT}\cC),
  $$
where $\cI$ is the identity operator on the space $\Theta \mS_{n+\nu}$.
The resolvents  $(\cI \pm \tau \ad_{\cA})^{-1}$ are well-defined since the spectrum of the operator $\ad_{\cA}$ on $\Theta \mS_{n+\nu}$ is contained in the strip $\big\{z \in \mC:\, |\Re z|< \frac{1}{\tau}\big\}$ due to the $\tau$-ad\-mi\-ssi\-bi\-li\-ty (\ref{taugood}).

\begin{lem}
\label{lem:LM}
The
optimal coupling matrix $L$ in (\ref{L}) can be expressed in terms of the matrices $P$ and $Q$ from (\ref{PQ}) as
\begin{align}
\label{Lopt}
    L
    & =
    \frac{2}{\lambda}
    \Pi^{-1}
    [Q,P]_{12}
    P_{22}^{-1},
\end{align}
provided $\cP_{22}\succ 0$. Furthermore, the optimal energy matrix $M$ of the observer satisfies
\begin{equation}
\label{eq12}
    \frac{1}{2}\Big(\frac{1}{\tau}[\Sigma \Theta^{-1}, Q]_{12}+[\Theta \cC^{\rT}\cC, P]_{12}\Big)
    + [Q,P]_{11}\Theta_1 L -
    \Theta_1
    K[Q,P]_{12}
      + [Q,P]_{12}\Theta_2 M = 0.
\end{equation}
\end{lem}
\begin{proof}
The representation (\ref{Lopt}) follows directly from the first optimality condition in (\ref{optlie1_optlie2}) under the assumption $\cP_{22}\succ 0$. In order to establish (\ref{eq12}), we note that the  left-hand sides of (\ref{cPQALElie1_cPQALElie2}), (\ref{optlie1_optlie2}) involve pairwise commutators of the Hamiltonian matrices $\cA, P, Q \in \Theta \mS_{n+\nu}$. Application of the Jacobi identity \cite{D_2006} and the antisymmetry of the commutator leads to the relations
\begin{align}
\nonumber
0   & =
[[P,\cA],Q] +[[\cA,Q], P]+[[Q,P], \cA]\\
\nonumber
& =
\frac{1}{\tau}
[\Sigma\Theta^{-1}-P,Q]
+
\big[\Theta \cC^{\rT}\cC-\frac{1}{\tau} Q, P\big]
 +[[Q,P], \cA]\\
\label{jacob}
& =  \frac{1}{\tau}
[\Sigma\Theta^{-1},Q]
+
[\Theta \cC^{\rT}\cC, P]+[[Q,P], \cA]
\end{align}
for any $\tau$-admissible observer (the optimality conditions (\ref{optlie1_optlie2}) have not  been used here). By substituting the matrix $\cA$ from (\ref{cA}) into the right-hand side of (\ref{jacob}) and considering the $(\cdot)_{12}$ block of the resulting Hamiltonian matrix, it follows that
\begin{equation}
\label{jacob12}
    \frac{1}{\tau}[\Sigma \Theta^{-1}, Q]_{12}+[\Theta \cC^{\rT}\cC, P]_{12}
    + 2\big([Q,P]_{11}\Theta_1 L +[Q,P]_{12}\Theta_2 M
    -
    \Theta_1 (
    K[Q,P]_{12}+L[Q,P]_{22})\big) = 0.
\end{equation}
If the second optimality condition  in (\ref{optlie1_optlie2}) is satisfied, then the corresponding term in (\ref{jacob12}) vanishes, which leads to (\ref{eq12}).
\end{proof}

As can be seen from the proof of Lemma~\ref{lem:LM}, the relation (\ref{eq12}) holds for any $\tau$-admissible stationary point of the CQF problem regardless of the assumption $\cP_{22}\succ 0$.   Furthermore, (\ref{eq12}) is a linear equation with respect to $M$. This allows the optimal observer energy matrix  $M$ to be  expressed in terms of $P$, $Q$ from (\ref{PQ}) in the case of equal plant and observer dimensions under nondegeneracy conditions considered in the next section.

\section{THE CASE OF EQUALLY DIMENSIONED  PLANT AND OBSERVER}\label{sec:same}

We will now consider observers which have the same dimension as the plant:  $\nu=n$. In this case, the  observer will be called \emph{nondegenerate} if the matrices  $P$ and  $Q$  from (\ref{PQ})  satisfy
\begin{equation}
\label{nondeg}
  \cP_{22}\succ 0,
  \qquad
  \det([Q,P]_{12})\ne 0.
\end{equation}
The results of Sections~\ref{sec:opt} and \ref{sec:lie} lead to the following necessary conditions of optimality for nondegenerate observers.

\begin{thm}
\label{th:LM}
Suppose the plant and observer dimensions are equal: $n=\nu$. Then for any nondegenerate observer, which  is a stationary point of the CQF problem (\ref{cZ})--(\ref{ES}) under the assumptions of Theorem~\ref{th:stat}, the coupling and energy matrices are related by (\ref{Lopt}) and by
\begin{equation}
\label{Mopt}
    M =
    \Theta_2^{-1}([Q,P]_{12})^{-1}\Big(\Theta_1 K[Q,P]_{12}
    - [Q,P]_{11}\Theta_1 L- \frac{1}{2}\Big(\frac{1}{\tau}[\Sigma \Theta^{-1}, Q]_{12}+[\Theta \cC^{\rT}\cC, P]_{12}\Big)\Big)
\end{equation}
to the matrices $P$, $Q$ from (\ref{PQ}) satisfying the ALEs (\ref{cPQALElie1_cPQALElie2}).
\end{thm}
\begin{proof}
The first of the conditions (\ref{nondeg}) makes the  previously obtained representation (\ref{Lopt}) applicable, which leads to a nonsingular coupling matrix $L$ in view of the second condition in (\ref{nondeg}). The latter allows (\ref{eq12}) to be 
uniquely solved for the observer energy matrix $M$ in the form (\ref{Mopt}). 
\end{proof}

In combination with the ALEs (\ref{cPALE}) and (\ref{cQALE}) (or their Lie-algebraic form (\ref{PQ})--(\ref{cPQALElie1_cPQALElie2}), the relations (\ref{Lopt}) and (\ref{Mopt}) of Theorem~\ref{th:LM} provide a set of algebraic equations for finding the matrices $L$ and $M$ of a nondegenerate observer which is a stationary point in the CQF problem (\ref{cZ})--(\ref{ES}).

\section{CONCLUSION}\label{sec:conc}

We have discussed the state-space and frequency-domain computation of discounted averages with exponentially decaying weights for moments of system variables in QHOs.  For a quantum plant and a quantum observer in the form of directly coupled  QHOs, we have considered a CQF problem of minimizing the discounted mean square value of the estimation error together with a penalty on the observer back-action.  First-order necessary conditions of optimality have been obtained for this problem in the form of two coupled ALEs. We have applied Lie-algebraic techniques  to representing this set of equations in a more explicit form. The existence and uniqueness of a solution for these coupled ALEs is a complicated open problem. A numerical solution of these equations can be based on a homotopy algorithm \cite{MB_1985} (see also \cite{VP_2010b}) and will be discussed  elsewhere.


\end{document}